\documentclass[a4paper,11pt]{article}
\usepackage{graphicx,bm, color}
\usepackage{amsmath,amsfonts,amssymb}
\newcommand{\bx}{\mathbf{x}}
\newcommand{\bk}{\mathbf{k}}
\newcommand{\eps}{\varepsilon}
\newcommand{\omq}{\omega^{(\nu)}_{\mathbf{q}}}
\newcommand{\homq}{\hbar \, \omega^{(\nu)}_{\mathbf{q}}}
\newcommand{\nqn}{n^{(\nu)}_{\mathbf{q}}}
\newcommand{\setk}{\mathbb{R}^{2}}
\newcommand{\fdk}{{\mathbb{F}}(\bk, \mu)}
\newcommand{\fdkp}{{\mathbb{F}}(\bk', \mu)}
\newcommand{\sv}{\, ,}
\newcommand{\p}{\, .}
\newcommand{\dm}{\displaystyle}
\newcommand{\perogni}{\forall \,}
\newcommand{\eq}[1]{(\ref{#1})}
\newcommand{\fr}{\\[5pt]}
\newtheorem{theorem}{Theorem}
\newtheorem{lemma}{Lemma}
\newenvironment{proof}{{\bf Proof.}}{$\Box$}
\author{Armando Majorana
	\fr Department of Mathematics and Computer Science, \\
	University of Catania, Italy}
\title{}
\title{A BGK model for charge transport in graphene}
\begin{document}
\baselineskip=15pt
\maketitle
\begin{abstract}
The Boltzmann equation describes the detailed microscopic behaviour of a dilute gas, and represents the basis of the kinetic theory of gases. 
In order to reduce the difficulties in solving the Boltzmann equation, simple expressions of a collision operator have been proposed to replace the true Boltzmann integral term. These new equations are called kinetic models.
The most popular and widely used kinetic model is the Bhatnagar-Gross-Krook (BGK) model.
\\
In this work we propose and analyse a BGK model for charge transport in graphene.
\end{abstract}
MSC-class: 82C40 - 82C70 (Primary) 82D37 (Secondary)
%
%



\section{Introduction}

Graphene is a gapless semiconductor made of a sheet composed of a single layer of carbon 
atoms  arranged into a honeycomb hexagonal lattice \cite{CaNe}. In view of application in 
graphene-based electron devices, it is crucial to understand the basic transport properties of this material.  
In a semiclassical kinetic setting, the charge transport in graphene is described by four 
Boltzmann equations, one for electrons in the valence ($\pi$) band and one for electrons in the conductions ($\pi^*$) band, that in turn can belong to the $K$ or $K'$ valley.
In this paper we study the case of a single distribution function $f$ for electrons belonging to a conduction band. This corresponds to a physical case, where a n-type doping or equivalently a high value of the Fermi potential is considered, and the electrons, belonging to a conduction band, do not move to the valence band. Moreover $K$ and $K'$ are considered equivalent.
Under these assumptions Boltzmann equation writes
\begin{equation}
\frac{\partial f }{\partial t} + 
\dfrac{1}{\hbar} \, \nabla_{\mathbf{k}} \, \varepsilon 
\cdot \nabla_{{\mathbf{x}}} f  
- \frac{e}{\hbar} \, {\bf E} \cdot \nabla_{\mathbf{k}} f
= \left( \dfrac{df}{dt} \right)_{coll} ,
\label{BTE}
\end{equation}
where $\mathrm{t}$ is the time, $\mathbf{x}$ and $\mathbf{k} \in \setk$ are the position and the wave-vector of a charge particle, respectively.
We denote by $\nabla_{{\mathbf{x}}}$ and $\nabla_{\mathbf{k}}$ the gradients with respect to the position and wave-vector, respectively.
With a very good approximation \cite{CaNe} the energy bands is given by
$\varepsilon(\mathbf{k}) = \hbar \, v_F \left| \mathbf{k} \right| $,
where $v_F$ is the (constant) Fermi velocity, and $\hbar$ is the Planck constant divided by 
$2 \, \pi$.
The elementary (positive) charge is denoted by $e$, and ${\bf E}$ is the electric field. 
The collision operator of Equation~\eq{BTE} describes the interaction of electrons with acoustic, optical and $K$ phonons. \\
If the electric field ${\bf E}$ is not constant, then it must be self-consistently evaluated by coupling the Boltzmann equation \eq{BTE} with the Poisson equation for the electrostatic potential.

More recently accurate numerical solutions, based on the discontinuous Galerkin method, to the Boltzmann equation have been shown \cite{MajMaRo}, \cite{RoMajCo}. A comparison with 
results, obtained by using a direct simulation Monte Carlo approach, shows an excellent 
agreement, which gives a further validation of the numerical scheme. 

In this work the collision operator of Equation~(\ref{BTE}) is replaced by a relaxation (BGK) collision operator. We propose an operator, which retains the fundamental properties of the Boltzmann equation, such the mass conservation, the same equilibrium distribution functions and it properly deals with Pauli's exclusion principle in the degenerate case.
%
%
\section{The BGK model}
The collision operator of Equation~\eq{BTE} writes \cite{Toma}
\begin{equation}
\int_{\setk} S(\bk', \bk) \, f(t,\bx,\bk') \left[ 1 - f(t,\bx,\bk) \right] d \bk' -
\int_{\setk} S(\bk, \bk') \, f(t,\bx,\bk) \left[ 1 - f(t,\bx,\bk') \right] d \bk' \sv
\label{oper_coll}
\end{equation}
where the transition rate $S(\bk, \bk')$, related to electron-phonon scatterings, 
is described in detail in the following.
\\
The collision operator \eq{oper_coll} vanishes if $f$ is the Fermi-Dirac distribution
$$
\fdk = 
\dfrac{1}{1 + \exp \left( \dfrac{\eps(\bk) - \mu}{k_{B} \, T} \right)} ,
$$
where $k_B$ is the Boltzmann constant, $T$ the constant graphene lattice temperature, and $\mu$ is the chemical potential.
If the electric field $\bf E$ is null, then Fermi-Dirac distributions are solutions, which do not depend on time $t$ and space coordinates $\bx$, of the Boltzmann equation \eq{BTE}.
\\
In this paper we propose the following BGK collision operator
\begin{equation}
Q(f)(t,\bx,\bk) = 
\int_{\setk} S(\bk', \bk) \, \fdkp \left[ 1 - f(t,\bx,\bk) \right] d \bk' -
\int_{\setk} S(\bk, \bk') \, f(t,\bx,\bk) \left[ 1 - \fdkp \right] d \bk'   \p
\label{oper_bgk}
\end{equation}
It is derived from Equation \eq{oper_coll}, replacing the distribution $f(t,\bx,\bk')$, inside the integrals, with a Fermi-Dirac distribution.
Now the chemical potential $\mu$ becomes \emph{a new unknown}. 
Therefore we must add an equation to the kinetic model. It is
\begin{equation}
\int_{\setk} Q(f)(t,\bx,\bk) \: d \bk = 0 \p
\end{equation}
This equation guarantees the mass conservation, exactly as the same for the Boltzmann equation. 
\\
If we define
\begin{equation}
\Phi_{0}(\bk,\mu) = 
\int_{\setk} S(\bk', \bk) \, \fdkp \, d \bk'
\quad \mbox{ and } \quad
\Phi_{1}(\bk,\mu) = \int_{\setk} S(\bk, \bk') \left[ 1 - \fdkp \right] d \bk' \sv
\end{equation}
then \eq{oper_bgk} becomes
\begin{equation}
Q(f)(t,\bx,\bk) = 
\Phi_{0}(\bk,\mu) \left[ 1 - f(t,\bx,\bk) \right] - \Phi_{1}(\bk,\mu) f(t,\bx,\bk) \p
\end{equation}
It is evident that, for every $\bk$ and $\mu$,
\begin{eqnarray*}
&&
\Phi_{0}(\bk,\mu) \left[ 1 - \fdk \right] - \Phi_{1}(\bk,\mu) \fdk 
\\[3pt]
&&
= \int_{\setk} S(\bk', \bk) \, \fdkp \left[ 1 - \fdk \right] d \bk'   -
\int_{\setk} S(\bk, \bk') \, \fdk \left[ 1 - \fdkp \right] d \bk'
= 0 \p
\end{eqnarray*}
Using this identity, we have
\begin{eqnarray}
Q(f)(t,\bx,\bk) & = &\Phi_{0}(\bk,\mu) \left[ 1 - f(t,\bx,\bk) \right] - 
\Phi_{0}(\bk,\mu) \, \dfrac{1 - \fdk}{\fdk} f(t,\bx,\bk)
\nonumber
\\
& = &
\dfrac{\Phi_{0}(\bk,\mu)}{\fdk} \left\{
\fdk \left[ 1 - f(t,\bx,\bk) \right] -
\left[ 1 - \fdk \right] f(t,\bx,\bk) \right\}
\nonumber
\\
& = &
\dfrac{\Phi_{0}(\bk,\mu)}{\fdk} 
\left[ \fdk - f(t,\bx,\bk) \right] .
\end{eqnarray}
By defining the collision frequency of the BGK model
\begin{equation}
\kappa(\bk,\mu) = \dfrac{\Phi_{0}(\bk,\mu)}{\fdk} =
\Phi_{0}(\bk,\mu) 
\left[ 1 + \exp \left( \dfrac{\eps(\bk) - \mu}{k_{B} \, T} \right) \right] ,
\label{cfBGK}
\end{equation}
the kinetic model writes
\begin{eqnarray}
&&
\frac{\partial f }{\partial t}(t,\bx,\bk) + 
\dfrac{1}{\hbar} \, \nabla_{\mathbf{k}} \, \varepsilon 
\cdot \nabla_{{\mathbf{x}}} f(t,\bx,\bk)  
- \frac{e}{\hbar} \, {\bf E} \cdot \nabla_{\mathbf{k}} f(t,\bx,\bk)
= \kappa(\bk,\mu) \left[ \fdk - f(t,\bx,\bk) \right] ,
\label{BGKeq}
\\[5pt]
&&
\int_{\setk} \kappa(\bk,\mu) \left[ \fdk - f(t,\bx,\bk) \right] d \bk = 0 \p
\label{mass_cons}
\end{eqnarray}
A solution of system \eq{BGKeq}-\eq{mass_cons} consists of the two functions $f(t,\bx,\bk)$ and $\mu(t,\bx)$, which must satisfy suitable regularity conditions.
The common definitions of the kernel $S$ guarantees that the function $\kappa(\bk,\mu)$ is a no negative continuous function in the set $\setk \times \mathbb{R}$.
\subsection{The Pauli exclusion principle}
The distribution function $f$ must be not negative and less or equal to one, according to the Pauli's exclusion principle. We prove that, if the electric field is null, then every spatial homogeneous solution of Equations~\eq{BGKeq}-\eq{mass_cons} satisfies these conditions.
\begin{theorem}
Let be $f(t,\bk)$ a function defined in $[0, + \infty[ \times \setk$, differentiable with respect to $t$ for every $\bk \in \setk$, and integrable, with respect to $\bk$, over the domain $\setk$ for each time $t$. Moreover the function $\mu(t)$ is continuous in 
$[0, + \infty[$. 
If $f$ and $\mu$ satisfy the equations
$$
\frac{\partial f }{\partial t}(t,\bk)
= \kappa(\bk,\mu) \left[ \fdk - f(t,\bk) \right]
\quad \mbox{and} \quad
\int_{\setk} \kappa(\bk,\mu) \left[ \fdk - f(t,\bk) \right] d \bk = 0 \sv
$$
and $0 \leq f(0,\bk) \leq 1$, for all $\bk \in \setk$,
then $0 \leq f(t,\bk) \leq 1$ for $(t, \bk)  \in [0, + \infty[ \times \setk$.
\end{theorem}
\begin{proof}
Let be
$$
K(\bk,\mu) \quad \mbox{such that} \quad
\dfrac{\partial K}{\partial t} = \kappa(\bk,\mu) \quad \perogni t \geq 0 \p
$$
Therefore we have
$$
\exp(K(\bk,\mu)) \, \dfrac{\partial f}{\partial t} = 
\exp(K(\bk,\mu)) \, \kappa(\bk,\mu) \, \fdk - 
\exp(K(\bk,\mu)) \, \kappa(\bk,\mu) \, f(t, \bk) 
$$
that is
$$
\dfrac{\partial \mbox{ }}{\partial t} \left[ \exp(K(\bk,\mu)) \, f(t, \bk) \right] =
\exp(K(\bk,\mu)) \, \kappa(\bk,\mu) \, \fdk \geq 0
\p
$$
and then
$$
\exp(K(\bk,\mu)) \, f(t, \bk) \geq 
\left[ \dfrac{}{} \! \exp(K(\bk,\mu)) \, f(t, \bk) \, \right|_{t=0} \p
$$
Hence the distribution function $f(t, \bk)$ is always not negative, because 
$f(0, \bk) \geq 0$ for every $\bk \in \setk$. \\
Since $f(0, \bk) \leq 1$ for every $\bk \in \setk$, we define
$g(t, \bk) = 1 - f(t, \bk)$, and we note that $0 \le g(0,\bk) \leq 1$.
It is evident that $g$ satisfies the equation
$$
\dfrac{\partial g}{\partial t} = \kappa(\bk,\mu) 
\left[ \left( 1 - \fdk \right) - g \right]
$$
and then, as before, taking into account that $\fdk < 1$, from
$$
\exp(K(\bk,\mu)) \, g(t, \bk) \geq 
\left[ \dfrac{}{} \! \exp(K(\bk,\mu)) \, g(t, \bk) \, \right|_{t=0} \sv
$$
it follows 
$g(t, \bk) \geq 0 \Leftrightarrow f(t, \bk) \leq 1$.
\end{proof}
\subsection{The kernel of the Boltzmann collision operator}
The kernel $S(\bk, \bk')$ is given by the sum of terms of the kind
\begin{equation}
\left| G^{(\nu)}(\bk, \bk') \right|^{2}
\left[ \left( \nqn + 1 \right) 
\delta \left( \eps(\bk') - \eps(\bk) + \homq \right) 
+ \nqn 
\delta \left( \eps(\bk') - \eps(\bk) - \homq \right) 
\right] .
\label{Gkk}
\end{equation}
The index $\nu$ labels the $\nu$th phonon mode, $G^{(\nu)}(\bk', \bk)$ is the scattering rate, which describes the scattering mechanism between phonons $\nu$ and electrons.
The symbol $\delta$ denotes the Dirac distribution function, $\omq$ the $\nu$th phonon frequency, and $\nqn$ is the Bose-Einstein distribution for the phonon of type $\nu$
$$
\nqn = \dfrac{1}{e^{\homq /k_B T} - 1} \p
$$
For acoustic phonons, usually one considers the elastic approximation, and
$$
2 \, n^{(ac)}_{\mathbf{q}}
\left| G^{(ac)}(\bk', \bk) \right|^{2} = \dfrac{1}{(2 \, \pi)^{2}} \, 
\dfrac{\pi \, D_{ac}^{2} \, k_{B} \, T}{2 \hbar \, \sigma_m \, v_{p}^{2}}
\left( 1 + \cos \vartheta_{\bk \sv \bk'} \right) ,
$$
where $D_{ac}$ is the acoustic phonon coupling constant, $v_{p}$ is the sound speed in graphene, $\sigma_m$ the graphene areal density, and $\vartheta_{\bk \sv \bk'}$ is the convex angle between $\bk$ and ${\bk'}$.   
\\
There are three relevant optical phonon scatterings: the longitudinal optical (LO), the transversal optical (TO) and the ${K}$ (K) phonons.
The scattering rates are
\begin{eqnarray*}
\left| G^{(LO)}(\bk', \bk) \right|^{2} & = & 
\dfrac{1}{(2 \, \pi)^{2}} \, \dfrac{\pi \, D_{O}^{2}}{\sigma_m \, \omega_{O}}
\left( 1 - \cos ( \vartheta_{\bk \sv \bk' - \bk} + \vartheta_{\bk' \sv \bk' - \bk} ) \right)
\\
\left| G^{(TO)}(\bk', \bk) \right|^{2} & = & 
\dfrac{1}{(2 \, \pi)^{2}} \, \dfrac{\pi \, D_{O}^{2}}{ \sigma_m \, \omega_{O}}
\left( 1 + \cos ( \vartheta_{\bk \sv \bk' - \bk} + \vartheta_{\bk' \sv \bk' - \bk} ) \right)
\\
\left| G^{(K)}(\bk', \bk) \right|^{2} & = & 
\dfrac{1}{(2 \, \pi)^{2}} \, \dfrac{2 \pi \, D_{K}^{2}}{\sigma_m \, \omega_{K}}
\left( 1 - \cos \vartheta_{\bk \sv \bk'} \right) ,
\end{eqnarray*}
where $D_{O}$ is the optical phonon coupling constant, $\omega_{O}$ the optical phonon frequency, $D_{K}$ is the K-phonon coupling constant and $\omega_{K}$ the K-phonon frequency.
The angles $\vartheta_{\bk \sv \bk' - \bk}$ and $\vartheta_{\bk' \sv \bk' - \bk}$ denote 
the convex angles between $\bk$ and $\bk' - \bk$  and between $\bk'$ and  $\bk' - \bk$, 
respectively.
\\
We used the same physical parameters of the Table 1 of Ref.~\cite{RoMajCo}. 
\\
Since the phonon frequency of longitudinal and transversal optical phonons coincide, then, as we sum the corresponding terms \eq{Gkk}, the function 
$\cos ( \vartheta_{\bk \sv \bk' - \bk} + \vartheta_{\bk' \sv \bk' - \bk} )$ can be eliminated, easily.
Therefore, in this case, $S$ is a sum of terms \eq{Gkk} where now
$ \left| G^{(\nu)}(\bk, \bk') \right|^{2} =
{\cal G}^{(\nu)}(\cos \vartheta_{\bk \sv \bk'}) $.
\subsection{The collision operator of the BGK model}
Taking into account Equation~\eq{BGKeq} and the definition \eq{cfBGK}, the collision operator of the BGK model writes
\begin{equation}
\Phi_{0}(\bk,\mu) \left\{ 1 - f(t,\bx,\bk) 
\left[ 1 + \exp \left( \dfrac{\eps(\bk) - \mu}{k_{B} \, T} \right) \right]
\right\} .
\label{oper2_bgk}
\end{equation}
This expression can be simplified. 
To this aim, we consider the function $\Phi_{0}(\bk,\mu)$, which is given by the sum of the following integrals
\begin{align*}
&
\int_{\setk} 
{\cal G}^{(\nu)}(\cos \vartheta_{\bk \sv \bk'}) 
\left[ \left( \nqn + 1 \right) 
\delta \left( \eps(\bk') - \eps(\bk) - \homq \right) 
+ \nqn
\delta \left( \eps(\bk') - \eps(\bk) + \homq \right) 
 \right] \fdkp \: d \bk'
\\
&
\mbox{ } =
\dfrac{ \nqn + 1 }
{1 + \exp \left( \dfrac{\eps(\bk) + \homq - \mu}{k_{B} \, T} \right)}
\int_{\setk} 
{\cal G}^{(\nu)}(\cos \vartheta_{\bk \sv \bk'}) \,
\delta \left( \eps(\bk') - \eps(\bk) - \homq \right) d \bk'
\\
&
\mbox{} \quad +
\dfrac{ \nqn }
{1 + \exp \left( 
\dfrac{ \left[ \eps(\bk) - \homq \right]_{+} - \mu}{k_{B} \, T} 
\right)}
\int_{\setk} 
{\cal G}^{(\nu)}(\cos \vartheta_{\bk \sv \bk'}) \,
\delta \left( \eps(\bk') - \eps(\bk) + \homq \right) d \bk'
\sv
\end{align*}
where $[z]_{+} = \max \{ z, 0 \}$.
\\
Introducing polar coordinates
$ \bk = r ( \cos \vartheta, \sin \vartheta ) $ and
$ \bk' = r' ( \cos \vartheta', \sin \vartheta' ) $,
we have
\begin{align*}
&
\int_{\setk} 
{\cal G}^{(\nu)}(\cos \vartheta_{\bk \sv \bk'}) \,
\delta \left( \eps(\bk') - \eps(\bk) - \homq \right) d \bk'
\\
&
\mbox{ } =
\int_{0}^{+ \infty} d r' \int_{0}^{2 \, \pi} d \vartheta' \,
{\cal G}^{(\nu)}(\cos (\vartheta' - \vartheta) ) \,
\delta \left( \hbar \, v_F \, r' - \hbar \, v_F \, r - \homq \right) r'
\\
&
\mbox{} =
\left[ \int_{0}^{2 \, \pi} {\cal G}^{(\nu)}(\cos (\vartheta' - \vartheta) ) \:  
d \vartheta' \, \right]
\dfrac{1}{\hbar \, v_F} \int_{\mathbb{R}}
\delta \left( r' - r - \dfrac{\omq}{v_F} \right) r' \, H(r') \: d r'
\\
&
\mbox{} =
\left[ \int_{0}^{2 \, \pi} {\cal G}^{(\nu)}(\cos \vartheta' ) \:  d \vartheta' \right]
\dfrac{1}{\hbar \, v_F} 
\left( r + \dfrac{\omq}{v_F}  \right) 
H \left( r + \dfrac{\omq}{v_F}  \right)
=
\left[ \int_{0}^{2 \, \pi} {\cal G}^{(\nu)}(\cos \vartheta' ) \:  d \vartheta' \right]
\dfrac{\eps(\bk) + \homq}{(\hbar \, v_F)^{2}} \sv
\end{align*}
where $H$ is the Heaviside function.
Analogously
$$ 
\int_{\setk} 
{\cal G}^{(\nu)}(\cos \vartheta_{\bk \sv \bk'}) \,
\delta \left( \eps(\bk') - \eps(\bk) + \homq \right) d \bk'
=
\left[ \int_{0}^{2 \, \pi} {\cal G}^{(\nu)}(\cos \vartheta' ) \:  d \vartheta' \, \right]
\dfrac{\left[ \eps(\bk) - \homq \right]_{+}}{(\hbar \, v_F)^{2}} \p
$$
Hence
\begin{eqnarray*}
\Phi_{0}(\bk,\mu) & = & \sum_{\nu} \dfrac{1}{(\hbar \, v_F)^{2}}
\left[ \int_{0}^{2 \, \pi} {\cal G}^{(\nu)}(\cos \vartheta' ) \:  d \vartheta' \right]
\\
&&
\mbox{} \quad \times \left[
\dfrac{ \left( \nqn + 1 \right) 
\left( \eps(\bk) + \homq \right)}
{1 + \exp \left( \dfrac{\eps(\bk) + \homq - \mu}{k_{B} \, T} \right)}
+ 
\dfrac{ \nqn \left[ \eps(\bk) - \homq \right]_{+} }
{1 + \exp \left( \dfrac{ \left[ \eps(\bk) - \homq \right]_{+} - \mu}{k_{B} \, T} \right)} \right] .
\end{eqnarray*}
In order to simplify the notation, we define
\begin{equation}
\psi(\bk, \xi; a,b) =
\dfrac{ \left( a + 1 \right) \left( \eps(\bk) + b \right)}
{1 + \xi \exp \left( \dfrac{\eps(\bk) + b}{k_{B} \, T} \right)}
+ 
\dfrac{ a \left[ \eps(\bk) - b \right]_{+} }
{1 + \xi \exp \left( \dfrac{ \left[ \eps(\bk) - b \right]_{+} }{k_{B} \, T} \right)}.
\label{def_psi}
\end{equation}
So
\begin{equation}
\Phi_{0}(\bk,\mu) = \sum_{\nu} \dfrac{1}{(\hbar \, v_F)^{2}}
\left[ \int_{0}^{2 \, \pi} {\cal G}^{(\nu)}(\cos \vartheta' ) \:  d \vartheta' \right]
\psi \! \left( 
\bk, \exp \left( \dfrac{ - \, \mu}{k_{B} \, T} \right); \nqn, \homq \right) .
\end{equation}
Then the BGK collision operator \eq{oper2_bgk} is a linear combination, with constant positive coefficients, of the terms
\begin{equation}
\psi \! \left( 
\bk, \exp \left( \dfrac{ - \, \mu}{k_{B} \, T} \right); \nqn, \homq \right)  \left\{ 1 - f(t,\bx,\bk) 
\left[ 1 + \exp \left( \dfrac{\eps(\bk) - \mu}{k_{B} \, T} \right) \right]
\right\} .
\label{psi_nu}
\end{equation}
Since $\nqn \geq 0$ and $\homq \geq 0$, then we can consider the function $\psi$, defined in \eq{def_psi}, only for $a \geq 0$, $b \geq 0$, and $\xi > 0$.
\begin{lemma}
\label{monot}
If $0 \leq f(t,\bx,\bk) \leq 1$, then the collision operator \eq{oper2_bgk}	is a strictly increasing function of $\mu$ on $\mathbb{R}$.
\end{lemma}
\begin{proof}
Since Equation~\eq{oper2_bgk} is a linear combination, with constant positive coefficients, of the functions \eq{psi_nu}, it is sufficient to prove that every function \eq{psi_nu} is  strictly increasing on $\mathbb{R}$ with respect to the variable  $\mu$.
\\
Again to simplify the notation, we consider the general application
\begin{equation}
\lambda(\bk, \xi; a,b) = \psi(\bk, \xi; a,b)  \left\{ 1 - \varphi(\bk) 
\left[ 1 + \xi \exp \left( \dfrac{\eps(\bk)}{k_{B} \, T} \right) \right]
\right\} , \label{def_lam}
\end{equation}
where $0 \leq \varphi(\bk) \leq 1$ for all $\bk \in \setk$.
Equation \eq{def_lam} gives the expression \eq{psi_nu}, choosing the parameters $a$, $b$, $\xi$ appropriately, and $\varphi(\bk) = f(t,\bx,\bk)$, for fixed $t$ and $\bx$.
If we define
$$
w(\bk) = \exp \left( \dfrac{\eps(\bk)}{k_{B} \, T} \right)  
\quad \mbox{and} \quad
\beta = \exp \left( \dfrac{b}{k_{B} \, T} \right) ,
$$
then 
$$
\psi(\bk, \xi; a,b) =
\dfrac{ \left( a + 1 \right) \left( \eps(\bk) + b \right)}{1 + \xi \, \beta \, w(\bk)}
+ 
\dfrac{ a \left[ \eps(\bk) - b \right]_{+} }{1 + \xi  \, \beta^{-1} \, w(\bk)} \sv
$$
where we have taken into account the identity
$$
\dfrac{ a \left[ \eps(\bk) - b \right]_{+} }
{1 + \xi \exp \left( \dfrac{ \left[ \eps(\bk) - b \right]_{+} }{k_{B} \, T} \right)}
\equiv
\dfrac{ a \left[ \eps(\bk) - b \right]_{+} }
{1 + \xi \exp \left( \dfrac{ \eps(\bk) - b }{k_{B} \, T} \right)} .
$$
Hence we can write
\begin{eqnarray*}
\lambda(\bk, \xi; a,b) & = & 
\left[
\dfrac{ \left( a + 1 \right) \left( \eps(\bk) + b \right)}{1 + \xi \, \beta \, w(\bk)}
+ \dfrac{ a \left[ \eps(\bk) - b \right]_{+} }{1 + \xi  \, \beta^{-1} \, w(\bk)}
\right] 
\left[ 1 - \varphi(\bk) - \xi \, w(\bk) \, \varphi(\bk) \right]
\\
& = &
\left( a + 1 \right) \left( \eps(\bk) + b \right)
\left[ \dfrac{\beta^{-1} \, \varphi(\bk) + 1 - \varphi(\bk)}{1 + \xi \, \beta \, w(\bk)}
- \beta^{-1} \, \varphi(\bk) \right] 
\\
&&
\mbox{} +
a \left[ \eps(\bk) - b \right]_{+} 
\left[
\dfrac{\beta \, \varphi(\bk) + 1 - \varphi(\bk)}{1 + \xi  \, \beta^{-1} \, w(\bk)}
- \beta \, \varphi(\bk)
\right] , 
\end{eqnarray*}
which is strictly decreasing with respect to the positive variable $\xi$, because 
$0 \leq \varphi(\bk) \leq 1$. \\
Now the result follows for $\dm \xi = \exp \left( \dfrac{ - \, \mu}{k_{B} \, T} \right)$.
\end{proof}
\fr
It is useful, for the following, to establish some inequalities. 
From
$$
\beta^{-1} \leq \beta^{-1} \, \varphi(\bk) + 1 - \varphi(\bk) \leq 1
\quad \mbox{and} \quad
1 \leq \beta \, \varphi(\bk) + 1 - \varphi(\bk) \leq \beta \sv
$$
we obtain
\begin{align}
\lambda(\bk, \xi; a,b) \leq & 
\left( a + 1 \right) \left( \eps(\bk) + b \right)
\left[ \dfrac{1}{1 + \xi \, \beta \, w(\bk)} - \beta^{-1} \, \varphi(\bk) \right] 
\! +
a \left[ \eps(\bk) - b \right]_{+} 
\left[
\dfrac{\beta}{1 + \xi  \, \beta^{-1} \, w(\bk)} - \beta \, \varphi(\bk) \right] 
\nonumber
\\[3pt]
\leq & 
\left( a + 1 \right) \left( \eps(\bk) + b \right)
\left[ \dfrac{1}{\xi \, \beta \, w(\bk)} - \beta^{-1} \, \varphi(\bk) \right] 
+
a \left[ \eps(\bk) - b \right]_{+} 
\left[
\dfrac{\beta}{\xi  \, \beta^{-1} \, w(\bk)} - \beta \, \varphi(\bk) \right] ,
\label{lam_sup}
\end{align}
and
\begin{align}
\lambda(\bk, \xi; a,b) \geq & 
\left( a + 1 \right) \left( \eps(\bk) + b \right)
\left[ \dfrac{\beta^{-1}}{1 + \xi \, \beta \, w(\bk)} - \beta^{-1} \, \varphi(\bk) \right] 
\! +
a \left[ \eps(\bk) - b \right]_{+} 
\left[
\dfrac{1}{1 + \xi  \, \beta^{-1} \, w(\bk)} - \beta \, \varphi(\bk) \right] \nonumber
\\[3pt]
\geq & 
\beta^{-1} \,  \dfrac{\left( a + 1 \right) \left( \eps(\bk) + b \right)}
{1 + \xi \, \beta \, w(\bk)}  
-
\left[ \dfrac{\left( a + 1 \right) \left( \eps(\bk) + b \right)}{\beta}  +
\beta \, a \left[ \eps(\bk) - b \right]_{+} \right] \varphi(\bk) \p 
\label{lam_inf}
\end{align}
We remark that the variable $\xi$ has not been involved in inequalities.
\section{The mass conservation}
Equation~\eq{mass_cons}, that guarantees the conservation of mass, can be written as
$$
\int_{\setk} \Phi_{0}(\bk,\mu) \left\{ 1 - f(t,\bx,\bk) 
\left[ 1 + \exp \left( \dfrac{\eps(\bk) - \mu}{k_{B} \, T} \right) \right]
\right\} d \bk = 0 \sv
$$
i.e.
\begin{equation}
\sum_{\nu}
\left[ \int_{0}^{2 \, \pi} {\cal G}^{(\nu)}(\cos \vartheta' ) \:  d \vartheta' \right]
\int_{\setk}
\left[ \lambda \! \left(\bk, \exp \left( \dfrac{ - \, \mu}{k_{B} \, T} \right); \nqn, \homq \right) \right]_{\varphi(\bk) = f(t,\bx,\bk)} d \bk = 0 \p
\label{eq_mu}
\end{equation}
\begin{theorem}
If $ 0 \leq f(t,\bx,\bk) \leq 1$ and $\eps(\bk) \, f(t,\bx,\bk)$ is integrable with respect to $\bk$ over $\setk$, for all $t$ and $\bx$, then there exists a unique $\mu$ satisfying Equation~\eq{eq_mu}.
\end{theorem}
\begin{proof}
If there exists a solution $\mu$ of Equation~\eq{eq_mu}, then it must be unique due to the Lemma \eq{monot}.
To prove the theorem, we show that every integrals of Equation~\eq{eq_mu} is negative for 
$\mu \rightarrow - \infty$, and positive for $\mu \rightarrow + \infty$.
We do not consider the meaningless case $f(t,\bx,\bk) = 0$ almost everywhere.
Taking into account the inequality \eq{lam_sup}, we have
\begin{align*}
& \int_{\setk}
\left[ \lambda(\bk, \xi; a, b) \right]_{\varphi(\bk) = f(t,\bx,\bk)} d \bk 
\\
& \mbox{} \leq  \int_{\setk} \left\{
\left( a + 1 \right) \left( \eps(\bk) + b \right)
\left[ \dfrac{1}{\xi \, \beta \, w(\bk)} - \beta^{-1} \, f(t,\bx,\bk) \right] 
+
a \left[ \eps(\bk) - b \right]_{+} 
\left[ \dfrac{\beta}{\xi  \, \beta^{-1} \, w(\bk)} - \beta \, f(t,\bx,\bk) \right]
\! \right\} d \bk
\\[5pt]
& \mbox{} =
\dfrac{1}{\xi}  \int_{\setk}
\left[ \dfrac{\left( a + 1 \right) \left( \eps(\bk) + b \right)}{\beta \, w(\bk)} +
\dfrac{a \left[ \eps(\bk) - b \right]_{+} \beta}{\beta^{-1} \, w(\bk)} \right] d \bk
\\[5pt]
& \mbox{} \quad  - \int_{\setk}
\left[ \left( a + 1 \right) \left( \eps(\bk) + b \right) \beta^{-1} +
a \left[ \eps(\bk) - b \right]_{+} \beta \right] f(t,\bx,\bk) \: d \bk \sv
\end{align*}
where the last integral is finite due to hypotheses of the theorem.
If $\mu \rightarrow - \infty$, then $\xi \rightarrow + \infty$, and 
\begin{align*}
&
\lim_{\mu \rightarrow - \infty} \int_{\setk}
\left[ \lambda(\bk, \xi; a, b) \right]_{\varphi(\bk) = f(t,\bx,\bk)} d \bk 
\\
& \mbox{}
=  - \int_{\setk}
\left[ \left( a + 1 \right) \left( \eps(\bk) + b \right) \beta^{-1} +
a \left[ \eps(\bk) - b \right]_{+} \beta \right] f(t,\bx,\bk) \: d \bk < 0 \p
\end{align*}
This implies that the l.h.s. of Equation~\eq{eq_mu} is negative for 
$\mu \rightarrow - \infty$. \\
Now we consider the inequality \eq{lam_inf}. We have
\begin{eqnarray*}
&& \int_{\setk}
\left[ \lambda(\bk, \xi; a, b) \right]_{\varphi(\bk) = f(t,\bx,\bk)} d \bk 
\\
&& \mbox{} \geq
\beta^{-1} \left( a + 1 \right) 
\int_{\setk} \dfrac{ \eps(\bk) + b}{1 + \xi \, \beta \, w(\bk)} \: d \bk 
- \int_{\setk}
\left[ \dfrac{\left( a + 1 \right) \left( \eps(\bk) + b \right)}{\beta}  +
\beta \, a \left[ \eps(\bk) - b \right]_{+} \right] f(t,\bx,\bk) \: d \bk \sv
\end{eqnarray*}
where both integrals exist.
Since, for $\mu > 0$, we have
\begin{eqnarray*}
&&
\int_{\setk} \dfrac{ \eps(\bk) + b}{1 + \xi \, \beta \, w(\bk)} \: d \bk 
= 
2 \, \pi \int_{0}^{+ \infty} \dfrac{\hbar \, v_F \, r + b}
{1 + \beta \, \exp \left( \dfrac{\hbar \, v_F \, r - \mu}{k_{B} \, T} \right)} \, r \: dr
\\
&& \mbox{} \quad =
\dfrac{2 \, \pi}{(\hbar \, v_F)^{2}} \int_{0}^{+ \infty} \dfrac{s + b}
{1 + \beta \, \exp \left( \dfrac{s - \mu}{k_{B} \, T} \right)} \, s \: ds
\geq
\dfrac{2 \, \pi}{(\hbar \, v_F)^{2}} \int_{0}^{\mu} \dfrac{s + b}
{1 + \beta \, \exp \left( \dfrac{s - \mu}{k_{B} \, T} \right)} \, s \: ds
\\
&& \mbox{} \quad \geq
\dfrac{2 \, \pi}{(\hbar \, v_F)^{2}} \int_{0}^{\mu} \dfrac{s + b}{1 + \beta} \, s \: ds
=
\dfrac{2 \, \pi}{(\hbar \, v_F)^{2} \left( 1 + \beta \right)}
\left( \dfrac{1}{3} \, \mu^{3} + \dfrac{1}{2} \, b \, \mu^{2} \right) .
\end{eqnarray*}
Therefore
$$
\lim_{\mu \rightarrow + \infty} \int_{\setk}
\left[ \lambda(\bk, \xi; a, b) \right]_{\varphi(\bk) = f(t,\bx,\bk)} d \bk 
=  + \infty \p
$$
This implies that the l.h.s. of Equation~\eq{eq_mu} is positive for 
$\mu \rightarrow + \infty$, and it concludes the proof. 
\end{proof}
This theorem establishes that the kinetic model \eq{BGKeq}-\eq{mass_cons} is correctly-set.
\section{Conclusions}
In this paper we propose a BGK model for charge transport in graphene. The collision operator of this model replaces the true non linear Boltzmann integral operator.
A further equation, for a new unknown, is added to the kinetic equation in order to guarantee the mass conservation. 
The model would furnish good results when the electric field is not strong, so that the distribution function remains near an equilibrium Fermi-Dirac distribution.
Moreover the simple collision term allows analytical investigations, which may be prohibitive for the Boltzmann equation.
A numerical scheme, based on a discontinuous Galerkin method, for finding approximate solutions to the model, does not seem more simple than the scheme used for solving the Boltzmann equation in Ref.~\cite{CoMaRo}, due to the non-linearity of the integral equation for the mass conservation.

%
%
%
\section*{Acknowledgements}
The author acknowledges the financial support provided by the project
"Modellistica, simulazione e ottimizzazione del trasporto di cariche in strutture
a bassa dimensionalità", (2016-2018) University of Catania. 
\end{document}